\documentclass[11pt,reqno]{amsart}
\usepackage{amssymb}
\usepackage{amsfonts}
\usepackage{amsmath}
\usepackage{stmaryrd}
\usepackage{physics}
\usepackage{braket}
\usepackage{dsfont}
\usepackage{bbold}
\usepackage{graphicx}
\usepackage{relsize}
\usepackage{makecell}
\usepackage[table,xcdraw]{xcolor}
\usepackage{enumerate}
\usepackage[pagebackref, colorlinks = true, linkcolor = blue, urlcolor  = blue, citecolor = red]{hyperref}
\usepackage[margin=1in]{geometry}
\usepackage{enumitem}
\usepackage{mathtools}
\usepackage{graphbox}
\usepackage{comment}
\usepackage{float}
\usepackage[font=scriptsize]{caption}

\usepackage[capitalize]{cleveref}

\renewcommand{\epsilon}{\varepsilon}

\newtheorem{theorem}{Theorem}[section]
\newtheorem{proposition}[theorem]{Proposition}
\newtheorem{corollary}[theorem]{Corollary}
\newtheorem{lemma}[theorem]{Lemma}

\newtheorem*{conjecture*}{Conjecture}

\theoremstyle{definition}
\newtheorem{definition}[theorem]{Definition}
\theoremstyle{definition}
\newtheorem{remark}[theorem]{Remark}
\theoremstyle{definition}
\newtheorem{example}[theorem]{Example}
\theoremstyle{definition}
\newtheorem{question}
[theorem]{Question}

\newcommand\vertarrowbox[3][6ex]{%
  \begin{array}[t]{@{}c@{}} #2 \\
  \left\uparrow\vcenter{\hrule height #1}\right.\kern-\nulldelimiterspace\\
  \makebox[0pt]{\scriptsize#3}
  \end{array}%
}

\definecolor{darkgreen}{rgb}{0,0.392,0}

\newcommand{\la}{\langle}
\newcommand{\ra}{\rangle}

\newcommand{\A}{\mathcal{A}}
\newcommand{\B}{\mathcal{B}}

\newcommand{\Le}{\mathcal{L}}

\newcommand{\Real}{\mathbb{R}}
\newcommand{\Comp}{\mathbb{C}}

\newcommand{\id}{\mathrm{id}}

\newcommand{\sch}{\mathsf{SN}}
\newcommand{\SR}{{\rm SR}}
\newcommand{\SN}{{\rm SN}}
\newcommand{\Sep}{\mathsf{Sep}}

\newcommand{\ppt}{\mathsf{PPT}}

\newcommand{\PO}{\mathcal{POS}}
\newcommand{\CP}{\mathcal{CP}}
\newcommand{\SP}{\mathcal{SP}}
\newcommand{\DSP}{\mathcal{DSP}}
\newcommand{\EB}{\mathcal{EB}}
\newcommand{\PPT}{\mathcal{PPT}}
\newcommand{\DEC}{\mathcal{DEC}}

\newcommand{\UND}{\mathcal{UND}}

\newcommand{\Ad}{{\rm Ad}}

\newcommand{\proj}[1]{|#1\rangle\langle #1|}

\newcommand{\vect}{{\rm vec\,}}

\author{Sang-Jun Park}
\email{sjpark@whu.edu.cn}
\address{School of Mathematics and Statistics, Wuhan University, 
430070, Wuhan, Hubei, China}

\title[Every PPT map has finite EB index] {Every PPT channel has finite entanglement breaking index}

\begin{document}

\begin{abstract}
We prove that every PPT linear map has finite entanglement breaking index, thereby establishing the eventual entanglement breaking property of PPT channels in full generality. Furthermore, by utilizing completely positive maps with low entanglement dimensionality, we show that a large family of PPT maps, which strictly contains the class of 2-superpositive maps, has entanglement breaking index bounded above by 3, uniformly in the dimension. In particular, these results provide evidence that the PPT-cubed conjecture may hold in full generality.
\end{abstract}

\maketitle

\tableofcontents

\section{Introduction}
{
In quantum information theory, entanglement breaking (EB) channels represent a fundamental threshold beyond which quantum communication becomes effectively classical: a quantum channel $\Phi$ is entanglement breaking if $(\operatorname{id}\otimes\Phi)(\rho)$ is separable for every bipartite state $\rho$, or equivalently if its Choi matrix is separable \cite{HSR03}. For a channel $\Phi:M_d\to M_d$, its entanglement breaking index is defined by
    $$n_{\rm EB}(\Phi):=\inf \{n\geq 1:\Phi^n\ {\rm is\ entanglement\ breaking}\},$$
with $n_{\rm EB}(\Phi)=\infty$ when no such $n$ exists. This notion was first introduced and systematically studied by Lami and Giovannetti \cite{LG15,LG16} as a measure of the lifetime of entanglement under repeated applications of the same noise. Operationally, if one half of an entangled state is subjected successively to $\Phi$, then $n_{\rm EB}(\Phi)$ is the first time at which the resulting dynamics necessarily destroys its entanglement with any external reference system. 

A particularly important instance of this problem is provided by channels with positive partial transpose (PPT), namely completely positive maps $\Phi$ for which $\top\circ\Phi$ is also completely positive. The \emph{PPT-squared conjecture}, proposed by Christandl \cite{PPTsq} and subsequently studied in various works, predicts that PPT noise destroys entanglement after only two successive uses:
\begin{center}
    $\Phi$ is PPT \quad $\implies$ \quad $\Phi^2$ is entanglement breaking,
\end{center}
or equivalently $n_{\rm EB}(\Phi)\leq2$. One may ask more generally whether the composition of two PPT maps is necessarily entanglement breaking, which turns out to be equivalent to the above weaker assertion \cite{CMHW19}. Besides its intrinsic relation to the geometry of PPT and separable cones, this problem has a natural interpretation through entanglement swapping \cite{CF17}: PPT-entangled states may be regarded as bound-entangled links of a quantum network, and composition of their associated maps corresponds to joining such links by an intermediate swapping operation. If the resulting map is entanglement breaking, no entanglement can be established between the terminal systems through this procedure, which is closely related to limitations on PPT-based quantum repeaters and long-distance quantum key distribution \cite{BCHW15,AL26}.

Although the PPT-squared conjecture remains substantially stronger, the weaker question of whether the EB index of a PPT map is finite has gradually become accessible. The first affirmative result in this direction was presented by Kennedy, Manor, and Paulsen \cite{KMP18} who showed that
the iterates of every unital or trace-preserving PPT map approach the
entanglement breaking cone asymptotically. Subsequently, Rahaman, Jaques, and
Paulsen \cite{RJP18} proved that every unital PPT channel has finite entanglement breaking
index. This result was later extended to any faithful PPT channels, and more generally to PPT
completely positive maps having full-rank left and right Perron-Frobenius
eigenmatrices, by Hanson, Rouz\'{e}, and Stilck Fran\c{c}a \cite{HRSF20}. On the side of the PPT-squared conjecture, the qutrit case $d=3$
was fully proved independently by Christandl, M\"uller-Hermes, and Wolf and by
Chen, Yang, and Tang \cite{CMHW19,CYT19}. Further positive
results cover Choi-type maps \cite{SN22}, random covariant models
\cite{NP26}, symplectic and Cartan-covariant maps \cite{Par26,Pru25}, and strengthened qutrit composition statements
\cite{AL26}. Finally, a recent work proves that every PPT composition through a
qubit intermediate algebra is entanglement breaking \cite{Henry26}.
}

The contribution of the present paper is twofold. First, we establish the finite entanglement breaking index property for \emph{arbitrary} PPT channels, removing the full-rank eigenmatrices assumptions required in previous results. Second, by combining low-rank Kraus factorizations with composition mechanisms through qubit and qutrit intermediate systems, we derive explicit dimension-independent bounds on the EB index under additional low-Schmidt-number assumptions.

\subsection{Eventual entanglement breaking property of PPT channels}

In our first main theorem, we remove all the above restrictions and show that every PPT quantum channel, and more generally every PPT linear map, has finite EB index.

\begin{theorem}\label{thm:intro-main}
Let $\Phi:M_d\to M_d$ be a completely positive and completely copositive linear map. Then
    $$n_{\rm EB}(\Phi)<\infty.$$
Equivalently, there is an integer $n_0=n_0(\Phi)$ such that $\Phi^n$ is entanglement breakingfor every $n\geq n_0$
\end{theorem}
As mentioned earlier, the full-support eigenmatrices case of Theorem~\ref{thm:intro-main} is derived from
\cite{HRSF20}.  The missing case is that in which a right or left Perron-Frobenius
eigenmatrix $\rho\in M_d$ is singular. As the key intermediate steps of the proof of Theorem \ref{thm:intro-main}, we first construct explicit decompositions of $\Phi$ into sum of two CP maps $\A,\B:M_d\to M_d$ where $\A$ and $\B$ are supported on the Hilbert subspaces ${\rm Ran}(\rho)$ and its orthogonal completement ${\rm Ran}(\rho)^{\perp}$, respectively (Theorem \ref{thm:Perron-splitting-full}). This in particular allows us to estimate $n_{\rm EB}(\Phi)$ explicitly in terms of $n_{\rm EB}(\A)$ and $n_{\rm EB}(\B)$. As another crucial step, we show that the entanglement dimensionality (i.e., Schmidt number \cite{TH00}) of the PPT map $\Phi$ is governed by the rank of the eigenstate $\rho$ (Theorem \ref{thm:singular-support-SP}):
    $$\Phi\in \PPT \implies  \SN(C_{\Phi})\leq \max(\rank \rho, d-\rank \rho).$$
In particular, when $\rho$ is singular, this implies that $\SN(C_{\Phi})\leq d-1$ and hence we can conclude the proof using the induction hypotheses.

\subsection{Absolute bounds on the EB index via low-entanglement decompositions}
Although Theorem \ref{thm:intro-main} guarantees that all PPT channels have finite EB index, it is still unknown whether there is a channel-independent bound of EB index. 
Our second result provides a wide family of PPT maps whose entanglement breaking index is absolutely bounded. First, let $\SP_k$ be the cone of $k$-superpositive
maps, equivalently CP maps whose Choi matrix has the Schmidt number at most $k$ \cite{SSZ09}. We further enlarge these classes by considering the cone sum
    $$\DSP_k:=\{\Psi_1+\top\circ \Psi_2:\Psi_1,\Psi_2\in \SP_k\}.$$
We first emphasize that $\DSP_k$ is nontrivial and may contain non-CP maps whenever $k\geq 2$. Indeed, we have the inclusion chain
    $$\EB = \DSP_1 = \SP_1 \subsetneq \DSP_{2}\subseteq \cdots \subseteq \DSP_{d-1}\subsetneq \DSP_d = \DEC.$$
The two strict inclusions above are verified in Example \ref{ex:DSPex}. For example, the reduction map 
    $$R(X):=\Tr(X)I_d-X, \quad X\in M_d$$
belongs to the cone $\DSP_2$ while $R$ is not even $2$-positive. We also refer to Example \ref{ex:DSP-Symp} for other nontrivial examples of PPT maps in $\DSP_2 \setminus \SP_2$. Thus, the class $\DSP_2$ already contains a substantial family of positive maps. 

Our results connecting these classes to the associated EB index are summarized as follows.

\begin{theorem}\label{thm:intro-low-bridge}
Let $\Phi:M_d\to M_d$ be a PPT map.
\begin{enumerate}
    \item If $\Phi\in \DSP_2$, then $n_{\rm EB}(\Phi)\leq 3$, i.e., $\Phi^3\in \EB$.
    
    \item If $\Phi\in \DSP_3$, then $n_{\rm EB}(\Phi)\leq 5$.
\end{enumerate}
\end{theorem}

We refer to Theorem \ref{thm:DSP-bridge} for stronger, fully compositional statements. In particular, we show the cone inclusion
\begin{equation} \label{eq:SP2-bridge}
    \PPT \circ \DSP_2 \circ \PPT \subseteq \EB,
\end{equation}
and part (1) of Theorem \ref{thm:intro-low-bridge} follows directly from this inclusion. We make two additional comments on the cones $\DSP_k$. First, the identity map $\id_d$ is does not belong to any $\DSP_k$ unless $k=d$ (Example \ref{ex:DSPex}). Thus, the PPT-squared conjecture does not follow from Eq. \eqref{eq:SP2-bridge} and remains open. Second, to the best of our knowledge, it is unknown whether the cone $\DSP_2$ contains all PPT maps when $d\geq 4$ while the inclusion $\PPT(3,3)\subseteq \SP(3,3)\subseteq \DSP_2(3,3)$ was derived in \cite{YLT16}. This inclusion is equivalent to the dual statement (Proposition \ref{prop:DSP_PPT_equiv})
\begin{equation} \label{eq:2biPTDec}
    \Phi \text{ is $2$-positive and $2$-copositive} \stackrel{?}{\implies}  \Phi\in \DEC,
\end{equation}
We remark that a stronger open question posed in \cite{DMS23} asks whether every 
$2$-entanglement breaking map is decomposable. Therefore, Eq. \eqref{eq:2biPTDec} is a natural intermediate question whose affirmative answer would directly imply the the PPT-cubed conjecture.

\subsection{Organization of the paper}

{
Section~\ref{sec:prelim} review various notions and preliminary results in entanglement theory. Section~\ref{sec:splitting} provides a general method for decomposing a PPT map into the sum of two CP maps whose EB indices govern the EB index of the original map, in terms of the associated Perron–Frobenius eigenstate. By utilizing this decomposition we present the proof of Theorem \ref{thm:intro-main} in  Section~\ref{sec:eventual}. Finally, Section~\ref{sec:bridges} develops EB-composability through the cones $\DSP_k$ and proves the absolute three- and five-step bounds Theorem \ref{thm:intro-low-bridge}.
}

\section{Preliminaries on quantum entanglement} \label{sec:prelim}

In this section, we briefly review basic notions in quantum entanglement theory used throughout this paper. Let $M_d:=M_d(\Comp)$ denote the algebra of $d\times d$ complex matrices. We use Dirac's \textit{bra-ket} notations: column vectors $v \in \mathbb{C}^d$ are written as kets $\ket{v}$ and their conjugate transpose $v^* \in (\Comp^d)^*$ are written as bras $\bra{v}$.
The standard \emph{inner product} is denoted by $\langle v | w \rangle:=v^*w\in \Comp$ and the \emph{outer product} (rank-one operator) by $|v\ra\la w|:=vw^*\in M_d$. The standard basis in $\Comp^d$ is denoted by $\{\ket{j}:=e_j\}_{j=1}^d$, and we write $\ket{ij} := \ket{i} \otimes \ket{j}$ to denote tensor products of standard basis vectors. We allow \emph{unnormalized} quantum states: a \emph{state} is simply a positive operator $\rho\in B(H)^+$ on a finite-dimensional Hilbert space $H$. The unit-trace condition is immaterial for the homogeneous cone-theoretic definitions below.

Let $H_A=\Comp^{d_A}$ and $H_B=\Comp^{d_B}$ be finite-dimensional Hilbert spaces. Any nonzero bipartite vector $\xi\in H_{AB}:=H_A\otimes H_B$ admits a \textit{Schmidt decomposition} \cite{NiCh} 
    $$|\xi\ra=\sum_{i=1}^k{\lambda_i}|v_i\ra\otimes |w_i\ra$$
where $\lambda_1\geq \cdots \geq \lambda_k>0$, and $\{v_i\}_{i=1}^k$ and $\{w_i\}_{i=1}^k$ are orthonormal subsets in $H_A$ and $H_B$, respectively. The integer $k$ and the coefficients $\{\lambda_i\}_{i=1}^k$ are uniquely determined, and we call $k$ the \textit{Schmidt rank} of $\xi$ and write ${\rm SR}(|\xi\ra)=k$. For a natural number $k$, define
    $$\sch_k(d_A\otimes d_B):={\rm conv}\{|\xi\ra\la \xi|: \xi\in H_{AB},\;\; {\rm SR}(|\xi\ra)\leq k\}\subset B(H_{AB})^+$$
and write simply $\sch_k$ when no confusion arises. Then the \textit{Schmidt number} of a quantum state $\rho_{AB}\in (M_{d_A}\otimes M_{d_B})^+$ is defined as the least $k$ for which $\rho_{AB}\in \sch_k$, and we denote it by ${\rm SN}(\rho_{AB})=k$. By definition,
    $$\sch_k=\{\rho\in B(H_{AB})^+:SN(\rho)\leq k\},$$
and $\sch_1\subset \sch_2\subset \sch_3\subset\cdots$. Note that $1\leq \SN(\rho_{AB})\leq \min(d_A,d_B)$ and $\sch_k=B(H_{AB})^+$ whenever $k\geq \min(d_A,d_B)$. The cone $\sch_1(d_A\otimes d_B)$ is precisely the cone $\Sep(d_A\otimes d_B)$ of \textit{separable} states; states outside the set $\Sep$ are called \textit{entangled}.

Another important cone is $\ppt(d_A\otimes d_B)$, consisting of the states with \textit{positive-partial-trasnpose (PPT)}, namely the operators $\rho_{AB}\in B(H_{AB})^+$ satisfying
    $$\rho_{AB}^{\Gamma}:=(\id_A\otimes \top_B)(\rho_{AB})\geq 0.$$
The Perez-Horodecki- criterion \cite{Per96,HHH96} is reformulated as $\Sep\subseteq \ppt$, providing a simple yet still powerful necessary condition for separability. In particular, in $2\otimes2$ and $2\otimes3$ systems, this condition is also sufficient for separability \cite{stormer1963positive,Wor76}.

\medskip

A linear map $\Phi:M_{d_A}\to M_{d_B}$ is called \emph{positive} if $\Phi(M_{d_A}^+)\subseteq M_{d_B}^+$. For an integer $k\geq 1$, $\Phi$ is called \emph{$k$-positive} if its ampliation $\id_k\otimes \Phi:M_k\otimes M_{d_A}\to M_k\otimes M_{d_B}$ is positive. It is \emph{completely positive} (CP) if it is $k$-positive for every $k$, and \emph{completely copositive} (co-CP) if $\top\circ\Phi$ is CP, where $\top:X\mapsto X^{\top}$ denotes matrix transposition. Finally, $\Phi$ is \emph{decomposable} if it is the sum of a CP map and a co-CP map. We denote by
    $$\PO(d_A,d_B), \quad \PO_k(d_A,d_B), \quad \CP(d_A,d_B), \quad \DEC(d_A,d_B) $$
the cones of positive, $k$-positive, CP, and decomposable maps from $M_{d_A}$ to $M_{d_B}$, respectively; we often omit the dimensions when they are clear. It is well-known that the complete positivity of $\Phi$ is equivalent to the $\min(d_A,d_B)$-positivity \cite{Cho75a}. 

The Choi--Jamio{\l}kowski isomorphism provides an efficient way to analyze the entanglement properties of linear maps \cite{Jam72,Cho75a}. 
For a linear map $\Phi:M_{d_A}\to M_{d_B}$, we use the (unnormalized) \emph{Choi matrix}
    $$C_{\Phi}=\sum_{i,j=1}^{d_A} |i\ra\la j|\otimes\Phi(|i\ra\la j|)
 \in M_{d_A}\otimes M_{d_B}.$$
Choi's theorem \cite{Cho75a} says that $\Phi$ is completely positive if and only if
$C_{\Phi}\geq 0$. 
We call $\Phi$ \emph{PPT} (resp. \emph{entanglement breaking; EB}) when $C_{\Phi}\in \ppt$ (resp. $C_{\Phi} \in \Sep$). Let us denote the corresponding cones of maps from $M_{d_A}$ to $M_{d_B}$ by $\PPT(d_A,d_B)$ and $\EB(d_A,d_B)$, respectively. Since $C_{\top \circ \Phi} = C_{\Phi}^{\Gamma}$, $\Phi$ is PPT if and only if $\Phi$ is CP and co-CP. These notions admit the equivalent characterizations \cite{HSR03,TextHol} 
\begin{align*}
    \Phi\in \PPT(d_A,d_B) \quad &\iff \quad \forall\,k\geq 1, \;\;(\id_k\otimes \Phi)((M_{k}\otimes M_{d_A})^+)\subseteq \ppt(k\otimes d_B),\\
    \Phi\in \EB(d_A,d_B) \quad &\iff \quad \forall\,k\geq 1, \;\;(\id_k\otimes \Phi)((M_{k}\otimes M_{d_A})^+)\subseteq \Sep(k\otimes d_B),
\end{align*}
which justifies the terminology. Furthermore, both $\PPT(d_A,d_B)$ and $\EB(d_A,d_B)$ are \emph{mapping cones} \cite{Sto86,GKS21,Kye23}, i.e., convex cones of the positive linear maps satisfying the $\CP$-bimodule property: for $\mathcal{C}=\PPT$ or $\EB$,
\begin{equation} \label{eq:MappingCone}
    \Psi_l\circ \Phi\circ \Psi_r\in \mathcal{C} \text{ whenever $\Phi\in \mathcal{C}$ and $\Psi_l,\Psi_r$ are CP}.
\end{equation}

We also recall the Schmidt-number mapping cones studied extensively in \cite{SSZ09}. We first write $\Ad_K: X\mapsto KXK^*$ the elementary CP map associated with a matrix $K$.
\begin{definition}\label{def:SP}
A CP map $\Phi:M_{d_A}\to M_{d_B}$ is called \emph{$k$-superpositive} if it
has a Kraus representation
    $$\Phi=\sum_\alpha\Ad_{K_\alpha}$$
for some $d_B\times d_A$ matrices $K_{\alpha}$ such that $\rank K_\alpha\leq k$. The corresponding cone is denoted by $\SP_k(d_A,d_B)$.
\end{definition}

Under vectorization, we have $C_{\Ad_K}=|\vect K^{\top}\ra\la \vect K^{\top}|$ where
    $$|\vect X\ra:= \sum_{i=1}^{d_B} (X|i\ra) \otimes |i\ra = \sum_{i,j} X_{ij}|ij\ra \in \Comp^{d_A}\otimes \Comp^{d_B}, \quad X=(X_{ij})_{1\leq i\leq d_A, 1\leq j\leq d_B}\in M_{d_A,d_B}.$$
Note that
$\SR(|\vect K^{\top}\ra)=\rank K^{\top} = \rank K$, and   therefore, $\Phi\in \SP_k$ if and only if it is CP and $\SN(C_{\Phi})\leq k$.
In particular, $\EB=\SP_1$ and we have further equivalent characterizations
\begin{align*}
    \Phi\in \SP_k(d_A,d_B) \quad &\iff \quad \forall\,n\geq 1, \;\;(\id_{n}\otimes \Phi)((M_{n}\otimes M_{d_A})^+)\subseteq \sch_k(n\otimes d_B).
\end{align*}
As advertised earlier, each $\SP_k$ is a mapping cone, that is, a convex cone satisfying Eq. \eqref{eq:MappingCone}. Furthermore, all the cones $\PPT, \EB$, and $\SP_k$ are invariant under Hilbert--Schmidt adjoints $\Phi\mapsto \Phi^*$ where $\Phi^*:M_{d_B}\to M_{d_A}$ is the linear map characterized by
    $$\Tr\big(Y^*\Phi(X)\big) = \Tr\big((\Phi^*(Y))^*X\big), \quad X\in M_{d_A},\;\; Y\in M_{d_B}.$$

For a CP map $\Phi:M_d\to M_d$, we define its \emph{entanglement breaking index} by
    $$n_{\rm EB}(\Phi) :=\inf\{n\geq 1:\Phi^n\in\EB\}.$$
If no such $n$ exists, set $n_{\rm EB}(\Phi)=\infty$. If $n_{\rm EB}<\infty$, then the mapping cone property implies that $\Phi^n\in \EB$ for every $n\geq n_{\rm EB}$, in which case $\Phi$ is called \emph{eventaully entanglement breaking}.

\medskip

For our main focus, we recall the following finite-dimensional Perron--Frobenius theory for positive maps; see, for example, \cite{EHK78,HRSF20} and \cite[Chapter 6]{Wol12}.
Every linear map $\Phi:M_d\to M_d$ has exactly $d^2$ complex eigenvalues (counting algebraic multiplicity), and the largest absolute value
    $$r_{\rm sp}(\Phi):=\max\{|\lambda|:\exists\, X\in M_d \text{ such that $X\neq 0$ and $\Phi(X)=\lambda X$}\}$$
is called the \emph{spectral radius} of $\Phi$. 

\begin{theorem}\label{thm:PF}
Let $\Phi:M_d\to M_d$ be a positive map and let
$\lambda=r_{\rm sp}(\Phi)$.  Then $\lambda$ is an eigenvalue of $\Phi$,
and there are nonzero positive semidefinite matrices $\rho,X\geq 0$ such that
    $$\Phi(\rho)=\lambda\rho, \quad \Phi^*(X)=\lambda X.$$
\end{theorem}

Most importantly, the following full-support case, which will be needed later, is derived in \cite[Theorem 3.14]{HRSF20}.

\begin{theorem}\label{thm:HRSF-full-support}
Let $\Phi:M_d\to M_d$ be a PPT map such that there are strictly
positive matrices $\rho,X>0$ satisfying
$$
 \Phi(\rho)=r_{\rm sp}(\Phi)\rho,
 \qquad
 \Phi^*(X)=r_{\rm sp}(\Phi)X.
$$
Then $n_{\rm EB}(\Phi)<\infty$.
\end{theorem}

{
The class considered in this paper is broader than the class of \emph{quantum channels}, that is, CP and trace-preserving (TP) maps. We note that a nonzero TP or unital positive map $\Phi$ cannot be \emph{nilpotent}, i.e., $\Phi^n\neq 0$ for any $n\geq 1$, as any such map has an eigenvalue 1 \cite{Wol12}. On the other hand, nilpotent PPT maps occur once normalization is dropped. This distinction is relevant to the zero-spectral-radius branch of the proof of Theorem~\ref{thm:intro-main}. In the following, we present two examples of PPT nilpotent maps which might be of independent interest.
}

\begin{example}
For $d\geq2$, let us consider
    $$\mathcal{N}(X):=\sum_{j=1}^{d-1}\la j+1|X|j+1\ra \,|j\ra\la j| =\sum_{j=1}^{d-1}\Ad_{|j\rangle\langle j+1|}(X), \quad X\in M_d$$
 
Every Kraus operator has rank one, so $\mathcal{N}_d$ is entanglement breaking. Moreover,
    $$\mathcal{N}^m(X) =\sum_{j=1}^{d-m}\langle j+m|X|j+m\rangle\, |j\ra\la j|,
 \quad 1\leq m\leq d,$$
and hence $\mathcal{N}^d=0$, while $\mathcal{N}^{d-1}(\proj{d})=\proj{1}\neq0$.
\end{example}

We also provide an example of PPT non-EB nilpotent map obtained from \emph{unextendible product basis}.

\begin{example} 
Let us consider the five real product vectors in $\Comp^3\otimes\Comp^3$
\begin{align*}
 |\psi_1\ra=|1\ra\otimes\frac{|1\ra-|2\ra}{\sqrt{2}}, \quad |\psi_2\ra&=\frac{|1\ra-|2\ra}{\sqrt{2}}\otimes|3\ra, \quad |\psi_3\ra=|3\ra\otimes\frac{|2\ra-|3\ra}{\sqrt{2}}, \quad |\psi_4\ra=\frac{|2\ra-|3\ra}{\sqrt{2}}\otimes|1\ra,\\
 |\psi_5\ra&=\frac{|1\ra+|2\ra+|3\ra}{\sqrt{3}}
 \otimes\frac{|1\ra+|2\ra+|3\ra}{\sqrt{3}}.
\end{align*}
They form the TILES unextendible product basis \cite{BDMSS99}, and it is shown that
    $$\rho_{\rm UPB} :=\frac{1}{4}\Big(I_3^{\otimes 2}-\sum_{j=1}^5|\psi_j\ra\la \psi_j|\Big)$$
is PPT and entangled . Let $\Pi:=|1\ra\la 1|+|2\ra\la 2| + |3\ra\la 3|\in M_6$ be the projection onto ${\rm span}\{|1\ra,|2\ra,|3\ra\}\subset\Comp^6$, let $\Pi^{\perp}=I_6-\Pi$, and choose coordinate isometries $V_{\Pi},V_{\Pi^{\perp}}:\Comp^3\to\Comp^6$ with ranges ${\rm Ran}(\Pi)$ and ${\rm Ran}(\Pi^{\perp})$, respectively. Now define a linear map $\Phi:M_6\to M_6$ by its Choi matrix
    $$C_{\Phi}=(V_{\Pi}\otimes V_{\Pi^{\perp}})\rho_{\rm UPB}(V_{\Pi}\otimes V_{\Pi^{\perp}})^*.$$
Then $C_{\Phi}\in \ppt\setminus \Sep$, and therefore, $\Phi$ is a PPT, non-EB map. Note that we have more precisely
    $$\Phi = \Ad_{V_{\Pi^{\perp}}} \circ \Phi_{\rm UPB} \circ \Ad_{V_{\Pi}^*},$$
where the Choi matix of $\Phi_{\rm UPB}:M_3\to M_3$ is defined to be $\rho_{\rm UPB}$. Since $V_{\Pi^{\perp}}^* V_{\Pi} = 0$, this implies
    $$\Phi^2 = \Ad_{V_{\Pi^{\perp}}} \circ \Phi_{\rm UPB} \circ (\Ad_{V_{\Pi}^*} \circ \Ad_{V_{\Pi^{\perp}}}) \circ \Phi_{\rm UPB} \circ \Ad_{V_{\Pi}^*} =0,$$
and hence $\Phi$ is nilpotent.
\end{example}

\section{Perron-Frobenius support splitting of PPT maps}\label{sec:splitting}

Throughout this section, let $\Phi:M_d\to M_d$ be a PPT map and suppose that there exists a nonzero positive semidefinite matrix $\rho\geq 0$ and a positive number $\lambda>0$ such that
    $$\Phi(\rho)=\lambda \rho.$$
By Theorem \ref{thm:PF}, the natural choice is the spectral radius $\lambda = r_{\rm sp}(\Phi)$, but the results presented in this section apply to any positive eigenvalue. Let us further take the projections onto the range $H_{\rho}:={\rm Ran}(\rho)$ of $\rho$ and its orthogonal complement
    $$\Pi_{\rho}={\rm supp\,}\rho:={\rm Proj}_{H_{\rho}}, \quad \Pi_{\rho}^{\perp}:=I_d-\Pi_{\rho}= {\rm Proj}_{H_{\rho}^{\perp}} = {\rm Proj}_{{\rm Ker}(\rho)}.$$
We mainly focus on the case $0\neq \Pi_{\rho}\neq I_d$, which is precisely the case not covered by Theorem \ref{thm:HRSF-full-support}.

We first note that $\Phi$ maps positive operators supported on $\Pi_{\rho}$ to positive operators supported on the same projection.

\begin{lemma}\label{lem:diagonal-support-propagation}
For every $u\in {\rm Ran}(\Pi_{\rho}) = H_{\rho}$,
    $$\Phi(\proj{u})=\Pi_{\rho} \Phi(\proj{u}) \Pi_{\rho}.$$
\end{lemma}

\begin{proof}
Because $\rho$ is strictly positive on the space $H_{\rho}$, there is a constant
$c_u>0$ such that
    $$\proj{u}\leq  c_u\rho.$$
Since $\Phi$ is positive, we have $\lambda c_u\rho - \Phi(|u\ra\la u|)\geq 0$ and thus $0\leq \Phi(|u\ra\la u|)\leq \lambda c_u\rho$. Now if $v\in H_{\rho}^{\perp} = {\rm Ker}(\rho)$, then one has
    $$0\leq \la v|\Phi(|u\ra\la u|)|v\ra \leq \lambda c_u \la v|\rho|v\ra = 0.$$
This shows $v\in {\rm Ker}(\Phi(|u\ra\la u|))$, and hence the conclusion follows.
\end{proof}

Note that the preceeding lemma uses only the positivity of $\Phi$. The PPT property gives a stronger conclusion for general rank-one operators $|u\ra\la w|$. To this end, we recall the following \emph{Schur-complement} criterion (see, for example, \cite[Chapter~1]{Zha05}).

\begin{lemma}
\label{lem:SchurComp}
A $2\times 2$ block matrix $M=\begin{pmatrix}P&Z\\ Z^*&Q\end{pmatrix}$
is positive semidefinite if and only if
    $$P\geq 0, \quad {\rm Ran}(Z)\subseteq {\rm Ran}(P), \quad \text{and} \quad Q-Z^*P^{-1}Z\geq 0,$$
where $P^{-1}$ denotes the \emph{Moore--Penrose inverse} of the positive semidefinite matrix $P$. 
\end{lemma}

\begin{proposition}\label{prop:two-sided-support}
For every $u\in {\rm Ran}(\Pi_{\rho}) = H_{\rho}$ and every $w\in\Comp^d$,
    $$\Phi(|u\ra\la w|)
 =\Pi_{\rho}\Phi(|u\ra\la w|)\Pi_{\rho}.$$
The same conclusion holds whenever $u\in \Comp^d$ and $w\in H_{\rho}$.
\end{proposition}

\begin{proof}
We only show the first assertion, as the second is analogous. Put $Z:=\Phi(\ket{u}\!\bra{w})$. Applying the positive map $\id_2\otimes \Phi$ to the rank-one operator
$\begin{pmatrix}
\proj{u}& |u\ra\la w|\\
|w\ra\la u|&\proj{w}
\end{pmatrix} = \begin{pmatrix}
u \\ w
\end{pmatrix}
\begin{pmatrix}
u \\ w
\end{pmatrix}^* \geq 0$ gives
    $$\begin{pmatrix}
 \Phi(\proj{u})& Z\\
 Z^*&\Phi(\proj{w})
 \end{pmatrix}\geq 0.$$
By Lemmas~\ref{lem:diagonal-support-propagation} and \ref{lem:SchurComp}, the block $\Phi(|u\ra\la u|)$ is
supported on $\Pi_{\rho}$ and
    $${\rm Ran}(Z)\subseteq {\rm Ran}(\Phi(|u\ra\la u|))\subseteq {\rm Ran}(\Pi_\rho).$$
This in particular implies $\Pi_{\rho}^{\perp} Z = 0$. On the other hand, the completely copositivity additionally implies that 
    $$\begin{pmatrix}
 \Phi(\proj{u})& Z^*\\
 Z&\Phi(\proj{w})
 \end{pmatrix} = (\top_2\otimes \Phi)\begin{pmatrix}
\proj{u}& |u\ra\la w|\\
|w\ra\la u|&\proj{w}
\end{pmatrix}\geq 0.$$
Therefore, we similarly have $\Pi_{\rho}^{\perp}Z^*=0$ and hence $Z\Pi_{\rho}^{\perp} = 0$.
Consequently, we obtain that
    $$Z = (\Pi_{\rho}+\Pi_{\rho}^{\perp})Z (\Pi_{\rho}+\Pi_{\rho}^{\perp}) =\Pi_{\rho}Z\Pi_{\rho}.$$
\end{proof}

\begin{corollary}\label{cor:linear-support-form}
If $X\in M_d$ satisfies $\Pi_{\rho}^{\perp}X\Pi_{\rho}^{\perp}=0$, then 
    $$\Phi(X)=\Pi_{\rho}\Phi(X)\Pi_{\rho}.$$
In particular, the set $\Pi_{\rho}M_d\Pi_{\rho}=\{\Pi_{\rho}X\Pi_{\rho}:X\in M_d\}$ is an invariant subspace of the map $\Phi$.
\end{corollary}

\begin{proof}
This follows from
Proposition~\ref{prop:two-sided-support} and the fact that every matrix satisfying $\Pi_{\rho}^{\perp}X\Pi_{\rho}^{\perp}=0$ is a linear combination of rank-one operators $|u\ra\la w|$ for which $u\in {\rm Ran}(\Pi_{\rho})$ or $w\in {\rm Ran}(\Pi_{\rho})$.
\end{proof}

Let us now decompose the Choi matrix $C_{\Phi}\in M_d^{\otimes 2}$ as a $2\times 2$ block matrix according to 
    $$(\Comp^d)^{\otimes 2}=(H_{\rho}\otimes \Comp^d)\oplus (H_{\rho}^{\perp}\otimes \Comp^d).$$
In other words, we write
\begin{equation}\label{eq:Choi-PQ-block}
 C_{\Phi}=
 \begin{pmatrix}
 P&Z\\Z^*&Q
 \end{pmatrix}\geq 0,
\end{equation}
where $P$ acts on $H_{\rho}\otimes \Comp^d$, $Q$ acts on
$H_{\rho}^{\perp}\otimes \Comp^d$, and $Z$ is the corresponding off-diagonal
operator. Then Corollary~\ref{cor:linear-support-form} implies the output support relations
\begin{equation} \label{eq:AB-output-support}
\begin{split}
    P &=(I_{H_{\rho}}\otimes \Pi_{\rho})P(I_{H_{\rho}}\otimes \Pi_{\rho}),\\
    Z &=(I_{H_{\rho}}\otimes \Pi_{\rho})Z(I_{H_{\rho}^{\perp}}\otimes \Pi_{\rho}).
\end{split}
\end{equation}
By applying Lemma~\ref{lem:SchurComp} to
\eqref{eq:Choi-PQ-block}, we obtain ${\rm Ran} Z\subseteq{\rm Ran} P$ and $Q-Z^*P^{-1}Z\geq 0.$ More precisely, we have the following positive decomposition
    $$C_{\Phi}= \underbrace{\begin{pmatrix} P&Z\\Z^*&Z^*P^{-1}Z
    \end{pmatrix}}_{\geq 0} + \underbrace{\begin{pmatrix} 0&0\\0&Q-Z^*P^{-1}Z
    \end{pmatrix}}_{\geq 0}.$$
Accordingly, we define two linear maps $\A,\B:M_d\to M_d$ by their Choi
matrices:
    $$C_{\A} = \begin{pmatrix} P&Z\\Z^*&Z^*P^{-1}Z
    \end{pmatrix}, \quad 
     C_{\B} =
    \begin{pmatrix} 0&0\\0&Q-Z^*P^{-1}Z
    \end{pmatrix}.$$
By definition, $\A$ and $\B$ are CP and satisfy $\Phi = \A + \B$. 

\begin{theorem}
\label{thm:Perron-splitting-full}
For every $X\in M_d$, the above maps $\A$ and $\B$ satisfy
    $$\mathcal{A}(X) =\Pi_{\rho}\A(X)\Pi_{\rho}, \quad  \mathcal{B}(X)=\B(\Pi_{\rho}^{\perp}X\Pi_{\rho}^{\perp}),$$
and $\B \circ \A = 0$.
\end{theorem}

\begin{proof}
The support relations in Eq. \eqref{eq:AB-output-support} imply that $P^{-1}$ is supported on
$H_{\rho}\otimes H_{\rho}$ and therefore
\begin{align*}
    Z^*P^{-1}Z &= \big((I_{H_{\rho}^{\perp}}\otimes \Pi_{\rho})Z^*(I_{H_{\rho}}\otimes \Pi_{\rho})\big)P^{-1}\big((I_{H_{\rho}}\otimes \Pi_{\rho})Z(I_{H_{\rho}^{\perp}}\otimes \Pi_{\rho})\big) \\
    &= (I_{H_{\rho}^{\perp}}\otimes \Pi_{\rho})Z^*P^{-1}Z(I_{H_{\rho}^{\perp}}\otimes \Pi_{\rho}).
\end{align*}
Consequently, one has
\begin{equation} \label{eq:Asupport}
    C_{\mathcal{A}} =(I_d\otimes \Pi_{\rho})C_{\mathcal{A}}(I_d\otimes \Pi_{\rho}),
\end{equation}
which is equivalent to $\mathcal{A}(\cdot) =\Pi_{\rho}\A(\cdot)\Pi_{\rho}$. On the other hand, since $C_{\B}$ is supported on the
$\Pi_{\rho}^{\perp}$-corner of the input Choi factor, 
$\B(|u\ra\la w|)=0$ whenever $u\in H_{\rho}$ or $w\in H_{\rho}$, which is exactly $\mathcal{B}(\cdot)=\B(\Pi_{\rho}^{\perp}\cdot \Pi_{\rho}^{\perp})$.  Finally, these two relations imply that
$$
 \B(\A(X))
 =\B(\Pi_{\rho}^{\perp}\Pi_{\rho}\A(X)\Pi_{\rho}\Pi_{\rho}^{\perp})=0,
$$
proving $\B\circ \A = 0$.
\end{proof}

We additionally define two linear maps $\Phi_{\rho}$ and $\Phi_{\rho^{\perp}}$ corresponding to two diagonal corners of the Choi matrix
\begin{align*}
    C_{\Phi_{\rho}}&:=C_{\Phi}\big|_{H_{\rho}\otimes H_{\rho}} = P\big|_{H_{\rho}\otimes H_{\rho}},\\
    C_{\Phi_{\rho^{\perp}}}&:=C_{\Phi}\big|_{H_{\rho}^{\perp}\otimes H_{\rho}^{\perp}} = Q\big|_{H_{\rho}^{\perp}\otimes H_{\rho}^{\perp}}.
\end{align*}
To understand these two maps explicitly, we consider two isometries $V_{\rho}:H_{\rho}\to \Comp^d$ and $V_{\rho^{\perp}}:H_{\rho}^{\perp}\to \Comp^d$ satisfying
$V_{\rho}V_{\rho}^*=\Pi_{\rho}$ and $V_{\rho^{\perp}}V_{\rho^{\perp}}^*=\Pi_{\rho}^{\perp}$, resp. Then we can write
\begin{equation} \label{eq:PPTCorner}
\begin{split}
    \Phi_{\rho} &=\Ad_{V_{\rho}^*}\circ\Phi\circ\Ad_{V_{\rho}}:B\big(H_{\rho}\big)\to B\big(H_{\rho}\big),\\
     \Phi_{\rho^{\perp}}&=\Ad_{V_{\rho^{\perp}}^*}\circ\Phi\circ\Ad_{V_{\rho^{\perp}}}:B\big(H_{\rho}^{\perp}\big)\to B\big(H_{\rho}^{\perp}\big).
\end{split}
\end{equation}
In particular, both $\Phi_{\rho}$ and $\Phi_{\rho^{\perp}}$ are PPT maps by the mapping cone property \eqref{eq:MappingCone}. 

The following factorizations relate the powers of $\A$ and $\B$ to those of the corner maps. It will then yield a recursive bound on the entanglement breaking index of $\Phi$.

\begin{proposition}
\label{prop:corner-factorization}
For every integer $p\geq 1$, we have
\begin{align*}
    \A^p &=\Ad_{V_{\rho}}\circ\Phi_{\rho}^{p-1}\circ \Ad_{V_{\rho}^*}\circ \A,\\
    \B^p &=\Ad_{V_{\rho^{\perp}}}\circ\Phi_{\rho^{\perp}}^{p-1}\circ\Ad_{(V_{\rho^{\perp}})^*} \circ \B.
\end{align*}
\end{proposition}

\begin{proof}
Theorem \ref{thm:Perron-splitting-full} first implies that
\begin{equation} \label{eq:ABsupport1}
    \A = \Ad_{\Pi_{\rho}}\circ \A = \Ad_{V_{\rho}}\circ \Ad_{V_{\rho}^*}\circ \A, \qquad \B = \B\circ \Ad_{\Pi_{\rho}^{\perp}} = \B \circ \Ad_{V_{\rho^{\perp}}}\circ \Ad_{(V_{\rho^{\perp}})^*}.
\end{equation}
Since $(V_{\rho^{\perp}})^*V_{\rho}=0$, the decomposition $\Phi=\A+\B$ further implies that
\begin{equation} \label{eq:ABsupport2}
\begin{split}
    \A\circ \Ad_{V_{\rho}} &= (\A+ \B \circ \Ad_{V_{\rho^{\perp}}}\circ \Ad_{(V_{\rho^{\perp}})^*}) \circ \Ad_{V_{\rho}} = \Phi \circ \Ad_{V_{\rho}},\\
    \Ad_{(V_{\rho^{\perp}})^*} \circ \B &= \Ad_{(V_{\rho^{\perp}})^*} \circ (\Ad_{V_{\rho}}\circ \Ad_{V_{\rho}^*}\circ \A + \B) = \Ad_{(V_{\rho^{\perp}})^*} \circ \Phi.
\end{split}
\end{equation}
Therefore, by combining Eqs. \eqref{eq:PPTCorner} and \eqref{eq:ABsupport2}, iteration of the first relation in \eqref{eq:ABsupport1} gives
    $$\A^p = \Ad_{V_{\rho}} \circ (\Ad_{V_{\rho}^*}\circ \A \circ \Ad_{V_{\rho}})^{p-1}\circ \Ad_{V_{\rho}^*} \circ\A = \Ad_{V_{\rho}} \circ \Phi_{\rho}^{p-1}\circ \Ad_{V_{\rho}^*} \circ\A,$$
showing the first identity. The second identity follows analogously.
\end{proof}

\begin{corollary}
\label{cor:recursive-EB-bound}
Let
$\Phi_{\rho}$ and $\Phi_{\rho^{\perp}}$ be the PPT corner maps defined as in Eq. \eqref{eq:PPTCorner}. Then we have
    $$n_{\rm EB}(\Phi) \leq n_{\rm EB}(\Phi_{\rho})+n_{\rm EB}(\Phi_{\rho^{\perp}})+1.$$
\end{corollary}
\begin{proof}
It suffices to consider the case in which $n_{\rm EB}(\Phi_{\rho}), n_{\rm EB}(\Phi_{\rho^{\perp}})<\infty$. Suppose $n_1,n_2\geq 1$ satisfies $\Phi_{\rho}^{n_1}\in \EB$ and $\Phi_{\rho^{\perp}}^{n_2}\in \EB$.
By Proposition~\ref{prop:corner-factorization}, one first has
\begin{align*}
    \A^{n_1+1} &=\Ad_{V_{\rho}}\circ\Phi_{\rho}^{n_1}\circ \Ad_{V_{\rho}^*}\circ \A\in \EB,\\
    \B^{n_2+1} &=\Ad_{V_{\rho}^{\perp}}\circ\Phi_{\rho^{\perp}}^{n_2}\circ\Ad_{(V_{\rho}^{\perp})^*} \circ \B \in \EB.
\end{align*}
On the other hand, the decomposition $\Phi=\A+\B$ and the relation $\B\circ \A$ from Theorem \ref{thm:Perron-splitting-full} imply
    $$\Phi^n = (\A+\B)^n = \sum_{j=0}^n \A^{n-j}\circ \B^j$$
for all $n\geq 1$. In particular, if $n=n_1+n_2+1$, then $\A^{n-j}\in \EB$ or $\B^{j}\in \EB$ for every $j$, and thus $\Phi^n\in \EB$. This shows that $n_{\rm EB}(\Phi)\leq n_1+n_2+1$, and taking $n_1=n_{\rm EB}(\Phi_{\rho})$ and $n_2=n_{\rm EB}(\Phi_{\rho^{\perp}})$ completes the proof.
\end{proof}

As another important consequence, we show that the entanglement dimensionality of $\Phi$ is governed by the rank of $\rho$.

\begin{theorem}
\label{thm:singular-support-SP}
The Schmidt number of $C_{\Phi}$ is bounded by the $\max(\rank \Pi_{\rho}, \rank \Pi_{\rho^{\perp}})$. In particular, if $\Pi_{\rho}\neq I_d$, then
    $$\Phi\in\SP_{\max\{\rank \rho, d-\rank \rho\}}
 \subseteq\SP_{d-1}.$$
\end{theorem}

\begin{proof}
Recall that the positive operator
$C_{\A}$ is supported on ${\rm Ran}(I_d\otimes \Pi_{\rho}) = \Comp^d\otimes H_{\rho}$ by Eq. \eqref{eq:Asupport}, and hence 
    $$\SN(C_{\A})\leq \dim H_{\rho} = \rank \rho.$$
Similarly, $C_{\B}$ is supported on ${\rm Ran}(\Pi_{\rho^{\perp}}\otimes I_d) = H_{\rho}^{\perp}\otimes \Comp^d$, and thus
    $$\SN(C_{\B})\leq \dim H_{\rho}^{\perp} = d-\rank \rho.$$
Therefore, we conclude that
    $$\SN(C_{\Phi}) = \SN(C_{\A}+C_{\B}) \leq \max \big(\SN(C_{\A}), \SN(C_{\B})\big) \leq \max (\rank \rho, d-\rank \rho).$$
\end{proof}

\section{Every PPT map is eventually entanglement breaking}
\label{sec:eventual}

This section is devoted to proving our first main result, Theorem~\ref{thm:intro-main}, by induction on the matrix dimension. Indeed, the singular-support case for the eigenstate $\rho$ is handled by the Perron-Frobenius
splitting from Section~\ref{sec:splitting}; the full-support case is exactly
Theorem~\ref{thm:HRSF-full-support}.

\begin{proof}[\textbf{Proof of Theorem \ref{thm:intro-main}}]
We proceed by induction on $d\geq 2$. The base case $d=2$ is clear, since every PPT map on $M_2$ is EB and hence $n_{\rm EB}(\Phi)=1$. Now assume that the conclusion holds on $M_{d'}$ for every $d'<d$, and let
$\Phi:M_d\to M_d$ be PPT. Let us take $\lambda=r_{\rm sp}(\Phi)\geq 0$.
We distinguish four cases.

\medskip
\noindent\textbf{Case 1: $\lambda=0$.}
As a linear operator on the $d^2$-dimensional vector space $M_d$,
$\Phi$ has spectrum $\{0\}$. Its characteristic polynomial is
therefore $t^{d^2}$, and the Cayley--Hamilton theorem gives
    $$\Phi^{d^2}=0\in \EB.$$

\medskip
\noindent\textbf{Case 2: a right Perron eigenmatrix is singular.}
Assume $\lambda>0$.  By Theorem~\ref{thm:PF}, there exists a nonzero state
$\rho\in M_d^+$ satisfying $\Phi(\rho)=\lambda\rho.$ If $\rho$ is singular, then
$\Pi_{\rho}={\rm supp\,}\rho$ is a nonzero proper projection with the range $H_{\rho}={\rm Ran}(\rho)\subsetneq \Comp^d$.  Furthermore, the corner maps
    $$\Phi_{\rho}:B(H_{\rho})\to B(H_{\rho}), \quad \Phi_{\rho^{\perp}}:B(H_{\rho})\to B(H_{\rho})$$
defined in Eq. \eqref{eq:PPTCorner} are PPT and act on matrix algebras of dimensions $\dim H_{\rho}, \dim H_{\rho}^{\perp}<d$.  Therefore, the induction hypothesis yields that $n_{\rm EB}(\Phi_{\rho})<\infty$ and $n_{\rm EB}(\Phi_{\rho^{\perp}})$, which further certifies $n_{\rm EB}(\Phi)\leq n_{\rm EB}(\Phi_{\rho})+n_{\rm EB}(\Phi_{\rho^{\perp}})+1<\infty$ by Corollary~\ref{cor:recursive-EB-bound}.

\medskip
\noindent\textbf{Case 3: the chosen right Perron eigenstate is strictly positive,
but a left Perron eigenmatrix is singular.}
By Theorem~\ref{thm:PF}, choose a nonzero matrix $X\geq 0$ with $\Phi^*(X)=\lambda X.$
If $X$ is singular, then it is a singular right Perron eigenmatrix of the
PPT map $\Phi^*$. Since $\Phi^n$ is EB if and only if $(\Phi^*)^n = (\Phi^n)^*$ is EB, the Case~2 applied to $\Phi^*$ yields that
    $$n_{\rm EB}(\Phi) = n_{\rm EB}(\Phi^*)<\infty.$$

\medskip
\noindent\textbf{Case 4: full left and right Perron support.}
This is precisely the case from Theorem \ref{thm:HRSF-full-support}.
\end{proof}

We record two consequences of Theorem \ref{thm:intro-main}. A straightforward implication is that the eventual EB property can be characterized eventual PPT property for all CP maps.

\begin{corollary} \label{cor:PPTInd}
For any CP map $\Phi:M_d\to M_d$, $\Phi$ is eventually EB if and only if $\Phi$ is eventually PPT, that is,
    $$n_{\rm PPT}(\Phi):=\inf \{n\geq 1: \Phi^n\in \PPT\}<\infty.$$
Furthermore, we have in this case,
    $$n_{\rm EB}(\Phi) \leq n_{\rm PPT}(\Phi)\, n_{\rm EB}(\Phi^{n_{\rm PPT}}).$$
\end{corollary}
\begin{proof}
Every entanglement breaking map is PPT, so eventual entanglement breaking implies eventual PPT. Conversely, put $m=n_{\rm PPT}(\Phi)$ and $\Psi=\Phi^m\in \PPT$. By Theorem~\ref{thm:intro-main}, $r:=n_{\rm EB}(\Psi)<\infty$. Hence $\Phi^{mr}=\Psi^r\in\EB$, which proves both the converse and the stated estimate.
\end{proof}

Furhtermore, Theorem \ref{thm:intro-main} also controls compositions that factor through a
possibly smaller intermediate algebra.

\begin{corollary}
\label{cor:rectangular-eventual}
Let $\Phi_1:M_d\to M_k$ and $\Phi_2:M_k\to M_d$
be two PPT maps. Then we have
    $$n_{\rm EB}(\Phi_2\circ \Phi_1)\leq 1+ n_{\rm EB}(\Phi_1\circ \Phi_2)<\infty.$$
\end{corollary}
\begin{proof}
This follows from Theorem \ref{thm:intro-main} and the iteration relation
    $$(\Phi_2\circ \Phi_1)^{n+1}=\Phi_2\circ (\Phi_1\circ \Phi_2)^n\circ \Phi_1, \qquad n\geq 0.$$
\end{proof}

Although we have shown that every PPT map has finite EB index, it remains unknown whether there exists a channel-independent bound for the EB index. A natural weaker question is the following:

\begin{question}
For each $d\geq 4$, does there exist finite number $n(d)\geq 2$ such that $\Phi^{n(d)}\in \EB$ for every PPT map $\Phi:M_d\to M_d$?
\end{question}

The above question asks whether the value
    $$\sup\{n_{\rm EB}(\Phi):\Phi\in \PPT(d,d)\}$$
is finite or not, which does not follow directly from Theorem \ref{thm:intro-main}. In the following section, we provide a large family of nontrivial PPT maps whose EB index is absolutely bounded by $3$, regardless of the dimensions.

\section{entanglement breaking-composability via low degree of entanglement}
\label{sec:bridges}

The recent work \cite{AL26} introduces the notion of \emph{EB-composability} as a framework for composition problems extending the PPT-squared conjecture. Let $\mathcal{C}_1\subseteq \PO(d_A,d)$ and $\mathcal{C}_2\subseteq \PO(d,d_B)$ be two mapping cones. Then the pair $(\mathcal{C}_2, \mathcal{C}_1)$ is called EB-composable if $\Phi_2\circ \Phi_1\in \EB(d_A,d_B)$ for any $\Phi_1\in \mathcal{C}_1$ and $\Phi_2\in \mathcal{C}_2$. The PPT-squared conjecture asks whether $(\PPT(d,d_B), \PPT(d_A,d))$ is EB-composable for all $d_A,d_B,d\geq 2$, which is further equivalent to the EB-composability of $(\PPT(d,d),\PPT(d,d))$ for every $d\geq 2$ \cite{CMHW19}. The qutrit case of the PPT-squared conjecture was established in \cite{CMHW19,CYT19}. More precisely, those proofs further derive EB-composable pairs
    $$(\EB_2(3,3),\EB_2(3,3)) \quad \text{and} \quad (\SP_2(3,3),\PPT(3,3)),$$
together with the inclusions $\PPT(3,3)\subsetneq \EB_2(3,3)$ and $\PPT(3,3)\subsetneq \SP_2(3,3))$, respectively in \cite{CMHW19} and \cite{CYT19}. Here
    $$\EB_2(d_A,d_B):=\{\Phi\in \PO(d_A,d_B):(\id_2\otimes \Phi)(\rho)\in \Sep(2\otimes d_B)\;\forall\,\rho\in (M_{2}\otimes M_{d_A})^+\}$$
denotes the mapping cone of \emph{$2$-entanglement breaking maps}. The second EB-composable pair above was further enlarged in \cite{AL26} to the pair $(\SP_2(3,3),\UND_1(3,3))$, where
    $$\UND_1(d_A,d_B):=\{\Phi\in \CP(d_A,d_B): \top\circ \Phi\in \PO_2(d_A,d_B)\}$$
denotes the cone of CP maps whose Choi matrices are $1$-copy undistillable \cite{HHH98}. We also recall the recent result of \cite{Henry26}, which settles the PPT-squared conjecture when the intermediate system is a qubit.

\begin{theorem} \label{thm:Hen26}
Let $\Phi_1:M_{d_A}\to M_2$ and $\Phi_2: M_2 \to M_{d_B}$ be two PPT linear maps. Then the composition 
    $$\Phi_2\circ \Phi_1:M_{d_A}\to M_{d_B}$$
is entanglement breaking. In other words, the pair $(\PPT(2,d_B), \PPT(d_A,2))$ is EB-composable for all $d_A,d_B\geq 2$.
\end{theorem}

In this section, we investigate EB-composability through \emph{multiple} mapping cones in order to relate it to the entanglement breaking index of PPT maps. A natural $n$-tuple extension of EB-composability is as follows.
\begin{definition}
Let $d_0,\ldots, d_n\geq 2$ and let $\mathcal{C}_i\in \PO(d_{i-1},d_{i})$ be a mapping cone for each $i=1,\ldots, n$. Then the tuple $(\mathcal{C}_n,\ldots, \mathcal{C}_1)$ is called \emph{EB-composable} if $\Phi_n\circ \cdots \circ\Phi_1\in \EB(d_0,d_n)$ whenever $\Phi_i\in \mathcal{C}_i$ for each $i$. In other words, we have
    $$\mathcal{C}_n\circ \cdots \circ \mathcal{C}_1:=\{\Phi_n\circ \cdots \circ \Phi_1:\Phi_i\in \mathcal{C}_i\;\forall\,i\}\subseteq \EB.$$
\end{definition}

For our purpose and for simplicity, we mainly consider the case $d_0=\cdots= d_n=d$ although many of the properties presented in this section hold true without this assumption. Let us first collect several basic properties, which follow directly from the definition and the mapping cone property.

\begin{proposition} \label{prop:DSPBasic}
Let $\mathcal{C}_n,\ldots, \mathcal{C}_1$ be mapping cones.
\begin{enumerate}
    \item If $(\mathcal{C}_n,\ldots, \mathcal{C}_1)$ is EB-composable, then the $(n-1)$-tuple $(\mathcal{C}_{n},\ldots, \mathcal{C}_{i+2},\widetilde{\mathcal{C}}_{i+1,i}, \mathcal{C}_{i-1},\ldots, \mathcal{C}_1)$ is EB-decomposable, where
        $$\widetilde{\mathcal{C}}_{i+1,i}:={\rm conv}\{\Phi_{i+1}\circ \Phi_i:\Phi_i\in \mathcal{C}_{i+1},\;\; \Phi_i\in \mathcal{C}_i\}.$$

    \item If $(\mathcal{C}_n,\ldots, \mathcal{C}_1)$ is EB-composable and if     $\id_d\in \mathcal{C}_i$ for some $i$, then $(\mathcal{C}_{n},\ldots, \mathcal{C}_{i+1}, \mathcal{C}_{i-1},\ldots, \mathcal{C}_1)$ is EB-composable.

    \item If $\mathcal{C}_i=\CP$ for some $i$, then $(\mathcal{C}_n,\ldots, \mathcal{C}_1)$ is EB-composable if and only if $(\mathcal{C}_{n},\ldots, \mathcal{C}_{i+1}, \mathcal{C}_{i-1},\ldots, \mathcal{C}_1)$ is EB-composable.

    \item If $(\mathcal{C}_n,\ldots, \mathcal{C}_1)$ is EB-composable, then for any CP map $\Phi\in \bigcap_{i=1}^n \mathcal{C}_i$, we have $n_{\rm EB}(\Phi)\leq n$.
\end{enumerate}
\end{proposition}

In the following, we present new nontrivial EB-composable tuples by extending Theorem \ref{thm:Hen26}. We first introduce a cone of positive maps which enlarges $\SP_k$: for integer $k\geq 1$, set
    $$\DSP_k := \SP_k + \top\circ \SP_k=\{\Psi_1+\top \circ\Psi_2: \Psi_1,\Psi_2\in \SP_k\}\subseteq \PO.$$
Both summands are stable under two-sided CP comosition, so $\DSP_k$ is a mapping cone. By definition, we have $\DSP_1=\SP_1=\EB$ and $\DSP_d = \CP + \top \circ \CP = \DEC$. Furthermore, we have an inclusion chain
    $$\EB=\DSP_1\subseteq \DSP_2 \subseteq \cdots \subseteq \DSP_d = \DEC.$$
We actually have strict inclusions $\DSP_1\subsetneq \DSP_2$ and $\DSP_{d-1}\subsetneq \DSP_d$ as in the following examples.

\begin{example} \label{ex:DSPex}
Let $d\geq 2$.
\begin{enumerate}
    \item The identity map $\id_d\in \CP$ does not belong to $\DSP_k$ unless $k=d$. Indeed, it is known that $\Real_+\id_d$ is an extremal ray in the positivity cone $\PO(d,d)$ \cite{YH05}. Therefore, if $\id_d=\Psi_1 +  \top \circ \Psi_2$ for some $\Psi_1,\Psi_2\in \SP_k$, then either $\id_d\in \SP_k$ or $\top\in \SP_k$, and only the former case is possible with the case $k=d$.

    \item The \emph{reduction map} \cite{HH99}
        $$R(X):=\Tr(X)I_d - X, \quad X\in M_d,$$
    is a positive map which is not $2$-positive, as shown in \cite{Tom85}. However, the transposed map $\top \circ R:X\mapsto \Tr(X)I_d-X^{\top}$ is a CP map (often refered to as the \emph{Werner-Holevo channel} \cite{WH02} after normalization) and further $2$-superpositive (see, e.g., \cite{Kye23lect,PY24}). This shows that
        $$R\in \DSP_2\setminus \DSP_1.$$
    Note in particular that the class $\DSP_2\subset \PO$ is already nontrivial whenever $d\geq 3$, which contains a positive map which is not even 2-positive while $\id_d\notin \DSP_2$ and $\DSP_2\supset \SP_2$.  
\end{enumerate}
\end{example}

We now present the main result of this section.

\begin{theorem}\label{thm:DSP-bridge}
For any $d\geq 2$,
\begin{enumerate}
    \item the triple $(\PPT, \DSP_2, \PPT)$ is EB-composable;

    \item the tuple $(\PPT, \DSP_3, \PPT, \DSP_3, \PPT)$ is EB-composable.
\end{enumerate}
\end{theorem}
\begin{proof}
By Definition \ref{def:SP}, any positive map $\Lambda\in \DSP_k$ admits a Kraus-type representation
\begin{equation} \label{eq:DSPKraus}
    \Lambda = \sum_{\alpha} \Ad_{K_{\alpha}} + \sum_{\beta} \top \circ \Ad_{L_{\beta}}
\end{equation}
where all matrices $K_{\alpha},L_{\beta}\in M_d$ have rank at most $k$. In particular, this further implies that $K_{\alpha} = K_{l,\alpha}^*K_{r,\alpha}$ and $L_{\beta} = L_{l,\beta}^*L_{r,\beta}$ for some $k\times d$ rectangular matrices $K_{l,\alpha}, K_{r,\alpha}, L_{l,\beta}, L_{r,\beta}\in M_{k,d}$.

Now for $\Phi_1,\Phi_2\in \PPT$ and for $\Lambda\in \DSP_2$ having the representation \eqref{eq:DSPKraus} with $k=2$, we can write
    $$\Phi_2\circ \Lambda \circ \Phi_1 = \sum_{\alpha}(\Phi_2\circ \Ad_{K_{l,\alpha}^*})\circ (\Ad_{K_{r,\alpha}}\circ \Phi_1) + \sum_{\beta}(\Phi_2 \circ \top \circ \Ad_{L_{l,\beta}^*})\circ (\Ad_{L_{r,\beta}}\circ \Phi_2).$$
Since $\Phi_1$ and $\Phi_2$ are PPT, all the summands above are compositions of PPT maps factoring through $M_2$. Therefore, Theorem \ref{thm:Hen26} implies that $\Phi_2\circ \Lambda \circ \Phi_1\in \EB$, which shows (1).

For the proof of (2), let us take $\Phi_1,\Phi_2,\Phi_3\in \PPT$ and $\Lambda_1,\Lambda_2\in \DSP_3$ associated with the $3\times d$ matrices $\{K_{l,\alpha}^{(1)},K_{r,\alpha}^{(1)}, L_{l,\beta}^{(1)},L_{r,\beta}^{(1)}\}_{\alpha,\beta}$ and $\{K_{l,\alpha'}^{(2)},K_{r,\alpha'}^{(2)}, L_{l,\beta'}^{(2)},L_{r,\beta'}^{(2)}\}_{\alpha',\beta'}$, resp. We shall show that the composition
    $$\Phi_3\circ \Lambda_2 \circ \Phi_2 \circ \Lambda_1 \circ \Phi_1$$
is entanglement breaking. Indeed, let us expand the sums
\begin{align*}
    \Phi_3\circ \Lambda_2 \circ \Phi_2 \circ \Lambda_1 \circ \Phi_1 &= \sum_{\alpha,\alpha'} (\Phi_3\circ \Ad_{K_{l,\alpha'}^{*(2)}})\circ (\Ad_{K_{r,\alpha'}^{(2)}}\circ \Phi_2 \circ \Ad_{K_{l,\alpha}^{*(1)}}) \circ (\Ad_{K_{r,\alpha}^{(1)}} \circ \Phi_1)\\
    &+ \sum_{\alpha,\beta'} (\Phi_3 \circ \top \circ \Ad_{L_{l,\beta'}^{*(2)}})\circ (\Ad_{L_{r,\beta'}^{(2)}}\circ \Phi_2 \circ \Ad_{K_{l,\alpha}^{*(1)}}) \circ (\Ad_{K_{r,\alpha}^{(1)}} \circ \Phi_1) \\
    &+ \sum_{\beta,\alpha'} (\Phi_3\circ \Ad_{K_{l,\alpha'}^{*(2)}})\circ (\Ad_{K_{r,\alpha'}^{(2)}}\circ \Phi_2 \circ \top \circ \Ad_{L_{l,\beta}^{*(1)}}) \circ (\Ad_{L_{r,\beta}^{(1)}} \circ \Phi_1)\\
    &+ \sum_{\beta,\beta'} (\Phi_3\circ \top \circ \Ad_{L_{l,\beta'}^{*(2)}})\circ (\Ad_{L_{r,\beta'}^{(2)}}\circ \Phi_2 \circ \top \circ \Ad_{L_{l,\beta}^{*(1)}}) \circ (\Ad_{L_{r,\beta}^{(1)}} \circ \Phi_1).
\end{align*}
Then 
all four types of the middle linear maps
    $$\Ad_{K_{r,\alpha'}^{(2)}}\circ \Phi_2 \circ \Ad_{K_{l,\alpha}^{*(1)}}, \;\; \Ad_{L_{r,\beta'}^{(2)}}\circ \Phi_2 \circ \Ad_{K_{l,\alpha}^{*(1)}}, \;\; \Ad_{K_{r,\alpha'}^{(2)}}\circ \Phi_2 \circ \top \circ \Ad_{L_{l,\beta}^{*(1)}}, \;\; \Ad_{L_{r,\beta'}^{(2)}}\circ \Phi_2 \circ \top \circ \Ad_{L_{l,\beta}^{*(1)}}$$
are PPT maps on $M_3$ and hence $2$-superpositve by \cite{YLT16,CYT19}. Therefore, the assertion (1) now implies that every summand is entangelment-breaking, and thus the assertion (2) follows.
\end{proof}

\begin{remark} \label{rem:DSPgeneral}
Theorem \ref{thm:DSP-bridge} holds even in the setting of distinct dimensions. That is, for any dimensions $d_0,d_1,d_2,d_3\geq 2$, the triple
    $$(\PPT(d_3,d_2), \DSP_2(d_2,d_1), \PPT(d_1,d_0))$$
is EB-composable, and the second assertion admits an analogous generalization.
\end{remark}

\begin{proof} [\textbf{Proof of Theorem \ref{thm:intro-low-bridge}}]
This is a direct application of Theorem \ref{thm:DSP-bridge} combined with Proposition \ref{prop:DSPBasic} (4).
\end{proof}

Theorem \ref{thm:DSP-bridge} provides further nontrivial classes of EB-composable tuples by taking the two cones
    $$\mathcal{C}_{k,l}:={\rm conv}(\PPT\circ \DSP_k), \quad \mathcal{C}_{k,r}:={\rm conv}(\DSP_k\circ \PPT).$$
Let us first note that, for every $k=1,\ldots, d$,
\begin{align}
    \EB\subseteq \mathcal{C}_{k,l} &= {\rm conv}(\PPT\circ \SP_k)\subseteq \SP_k\cap \PPT, \label{eq:DSPComp1} \\
    \EB\subseteq\mathcal{C}_{k,r} &= {\rm conv}(\SP_k \circ \PPT)\subseteq \SP_k\cap \PPT. \label{eq:DSPComp2}
\end{align}
Indeed, if $\Phi\in \PPT$ and if $\Lambda = \Psi_1+ \top \circ \Psi_2\in \DSP_k$ with $\Psi_1,\Psi_2\in \SP_k$, the composition
    $$\Phi\circ \Lambda = \Phi \circ \Psi_1 + (\Phi\circ \top)\circ \Psi_2$$
clearly belongs to $\PPT \circ \SP_k$ and also to both $\PPT$ and $\SP_k$ by the mapping cone property \eqref{eq:MappingCone}. This shows the inclusions 
    $${\rm conv}(\PPT\circ \SP_k) \subseteq \mathcal{C}_{2,l} \subseteq {\rm conv}(\PPT\circ \SP_k) \subseteq \SP_k\cap \PPT,$$
and we in particular have $\mathcal{C}_{2,l} = {\rm conv}(\PPT\circ \SP_k)$. On the other hand, any entanglement breaking map $\Phi\in \EB = \SP_1$ can be written as 
    $$\Phi = \sum_{\alpha}\Ad_{K_{\alpha}}, \quad K_{\alpha} = |v_{\alpha}\ra\la w_{\alpha}|\in M_d,\;\; \|w_{\alpha}\|=1.$$
Then clearly $\Ad_{K_{\alpha}} = \Ad_{K_{\alpha}}\circ \Ad_{|w_{\alpha}\ra\la w_{\alpha}|}\in \EB\circ \EB \subseteq \PPT \circ \SP_k$, and hence $\Phi\in \mathcal{C}_{k,l}$. This shows $\EB\subseteq \mathcal{C}_{k,l}$, and the other relations \eqref{eq:DSPComp2} also follows analogously.

Now the simple applications of Eqs. \eqref{eq:DSPComp1}, \eqref{eq:DSPComp2}, Proposition \ref{prop:DSPBasic} (1), (4), and Theorem \ref{thm:DSP-bridge} imply the following.

\begin{corollary} \label{cor:EBComp}
Let $d\geq 2$. Then
\begin{enumerate}
    \item both two pairs $(\mathcal{C}_{2,l},\PPT)$ and $(\PPT, \mathcal{C}_{2,r})$ are EB-composable. In particular, the PPT-squared conjecture holds within the class $\mathcal{C}_{2,l}\cup \mathcal{C}_{2,r}$:
        $$\Phi_1,\Phi_2\in {\rm conv}(\mathcal{C}_{2,l}\cup \mathcal{C}_{2,r})\subseteq \SP_2\cap \PPT \implies \Phi_2 \circ \Phi_1\in \EB.$$

    \item all the triples $(\mathcal{C}_{3,l}, \mathcal{C}_{3,l}, \PPT)$, $(\mathcal{C}_{3,l}, \PPT, \mathcal{C}_{3,r})$, and $(\PPT, \mathcal{C}_{3,r},\mathcal{C}_{3,r})$
    are EB-composable.
\end{enumerate}
\end{corollary}

Since every map $\Phi\in {\rm conv}(\mathcal{C}_{2,l}\cup \mathcal{C}_{2,r})$ is $2$-superpositive, PPT-squared conjecture still remains open for maps with higher entanglement dimensionality. One may ask whether the cone ${\rm conv}(\mathcal{C}_{2,l}\cup \mathcal{C}_{2,r})$ covers the whole cone $\SP_2\cap \PPT$ which would imply that PPT-squared conjecture holds for all $2$-superpositive PPT maps. We verify below that this is not the case for $d\geq 3$. 

\begin{proposition} \label{prop:DSPChoi}
For $d\geq 2$ and $k=1,\ldots, d$, we have $\Phi\in \mathcal{C}_{k,l}$ if and only if $\Phi^*\in \mathcal{C}_{k,r}$. Furthermore, we have the Choi matrix characterizations
\begin{align} 
    \Phi \in \mathcal{C}_{k,l} \quad &\iff \quad C_{\Phi}\in {\rm conv}\{\rho\in \ppt(d\otimes d): \rank \big((\id_d\otimes \Tr)(\rho)\big)\leq k\}, \label{eq:DSPChoi1}\\
    \Phi \in \mathcal{C}_{k,r} \quad &\iff \quad C_{\Phi}\in {\rm conv}\{\rho\in \ppt(d\otimes d): \rank \big((\Tr\,\otimes\, \id_d)(\rho)\big) \leq k\}. \label{eq:DSPChoi2}
\end{align}
\end{proposition}
\begin{proof}
The first assertion is clear from Eqs. \eqref{eq:DSPComp1}, \eqref{eq:DSPComp2}, and the fact that the cones $\SP_k$ and $\PPT$ are preserved under the adjoint $\Phi\mapsto \Phi^*$. On the other hand, note that
    $$C_{\Psi\circ \Ad_K} = (K^{\top}\otimes I_d)C_{\Psi}(\overline{K}\otimes I_d), \quad C_{\Ad_K\circ \Psi} = (I_d\otimes K)C_{\Psi}(I_d\otimes K^*)$$
for any linear map $\Psi$ and any matrix $K$. Then by taking partial trace operations, one has
\begin{align*}
    (\id_d\otimes \Tr)(C_{\Psi\circ \Ad_K}) = K^{\top}(\id_d\otimes \Tr)(C_{\Psi})\overline{K}, \quad (\Tr\otimes \id_d)(C_{\Ad_K\circ \Psi}) = K(\Tr\otimes \id_d)(C_{\Psi})K^*
\end{align*}
whose rank are both bounded above by $\rank K$. By convexity, this and the Choi--Jamio{\l}kowski correspondence imply the direction `$\Longrightarrow$' of Eqs. \eqref{eq:DSPChoi1} and \eqref{eq:DSPChoi2}. For the reverse direction `$\Longleftarrow$', suppose $\Phi$ is a PPT map such that $\rho_A:=\big((\id_d\otimes \Tr)(C_{\Phi})\big)$ satisfies $\rank \rho_A \leq k$. By taking the support projection $\Pi_{\rho_A}$ of $\rho_A$, we have $\rank \Pi_{\rho_A}\leq k$ and
    $$C_{\Phi} = (\Pi_{\rho_A}\otimes I_d)C_{\Phi}(\Pi_{\rho_A}\otimes I_d) = C_{\Phi\circ \Ad_{\Pi^{\top}_{\rho_A}}}.$$
Therefore, we have $\Phi = \Phi\circ \Ad_{\Pi^{\top}_{\rho_A}}\in \PPT \circ \SP_k$. The convexity again completes the proof of Eq. \eqref{eq:DSPComp1}. The other charaterization \eqref{eq:DSPComp2} follows similarly.
\end{proof}

\begin{example}
Here we provide comparisons between the cones 
    $$\EB(d,d)\subseteq \mathcal{C}_{2,l}, \,\mathcal{C}_{2,r}\subseteq{\rm conv}(\mathcal{C}_{2,l}\cup \mathcal{C}_{2,r}) \subseteq \SP_2(d,d) \cap \PPT(d,d).$$
\begin{enumerate}
    \item If $d=2$, $\SP_2\cap \PPT = \EB$ and hence all the cones above coincide.
    
    \item In \cite{CYT19,YLT16}, it is shown that every PPT map  on $M_3$ is $2$-superpositive, and hence $\SP_2(3,3)\cap \PPT(3,3) = \PPT(3,3)$. However, since $(\SP_2(3,3), \PPT(3,3))$ and $(\PPT(3,3), \SP_2(3,3))$ are already EB-composable \cite{CYT19}, Corollary \ref{cor:EBComp} (1) implies that
        $$\mathcal{C}_{2,l}(3,3) = \mathcal{C}_{2,r}(3,3) = \EB(3,3),$$
    and therefore,
        $${\rm conv}(\mathcal{C}_{2,l}(3,3)\cup \mathcal{C}_{2,r}(3,3)) = \EB(3,3)\subsetneq \PPT(3,3)=\SP_2(3,3)\cap \PPT(3,3).$$
    More generally, we have ${\rm conv}(\mathcal{C}_{2,l}(d,d)\cup \mathcal{C}_{2,r}(d,d))\subsetneq \SP_2(d,d)\cap \PPT(d,d)$ for all $d\geq 3$. Indeed, if we take any PPT non-EB map $\Phi\in \PPT(3,3)\setminus \EB(3,3)$, then its trivial embedding
        $$\widetilde{\Phi}=\Ad_{V}\circ \Phi \circ \Ad_{V^*}:M_d\to M_d,$$
    where $V=|1\ra\la 1|+|2\ra\la 2| + |3\ra\la 3|\in M_{d,3}$ is an isometry from $\Comp^3$ to $\Comp^d$, belongs to $\SP_2\cap \PPT$ while the fact $\Phi\notin {\rm conv}(\mathcal{C}_{2,l}(3,3)\cup \mathcal{C}_{2,r}(3,3))$ is again lifted to $\widetilde{\Phi}\notin {\rm conv}(\mathcal{C}_{2,l}(d,d)\cup \mathcal{C}_{2,r}(d,d))$.

    \item If $d\geq 4$, then $\EB\subsetneq \mathcal{C}_{2,l}, \mathcal{C}_{2,r}$ and $\mathcal{C}_{2,l}\neq \mathcal{C}_{2,r}$. Indeed, let $\rho\in (M_2\otimes M_4)^+$ be any $2\otimes 4$ PPT entangled state (see e.g. \cite{Wor76,Hor97}). Then we always have
        $$\rank (\id_2\otimes \Tr)(\rho)=2, \quad \rank(\Tr\otimes \id_4)(\rho)=4,$$
    since any bipartite state supported on a $2\otimes 3$ system is separable. Now if we take $\widetilde{\rho}\in (M_d\otimes M_d)^+$ by trivially embedding $\rho$ into the $d\otimes d$ system and if a linear map $\Phi:M_d\to M_d$ is defined by the condition $C_{\Phi}=\widetilde{\rho}$, then Proposition \ref{prop:DSPChoi} implies 
        $$\Phi\notin \EB, \quad \Phi\in \mathcal{C}_{2,l}, \quad \text{and} \quad \Phi\notin \mathcal{C}_{2,r}.$$
    Again by Proposition \ref{prop:DSPChoi}, we additionally have $\Phi^*\in \mathcal{C}_{2,r}\setminus \mathcal{C}_{2,l}$.
    
    Note in particular that, by Remark \ref{rem:DSPgeneral}, this provides nontrivial examples of PPT maps $\Phi\in \PPT\setminus \EB$ for which
    \begin{center}
        $\Phi\circ \Psi$ becomes entanglement breaking for any PPT maps $\Psi\in \PPT(d',d)$.
    \end{center}
\end{enumerate}
\end{example}

While ${\rm conv}(\mathcal{C}_{2,l}\cup \mathcal{C}_{2,r})$ is constrained in $\SP_2\cap \PPT$, the class $\DSP_2$ contains PPT maps with \emph{arbitrarily high degree of entanglement} as the dimension increases.

\begin{example} \label{ex:DSP-Symp}
For an even integer $d\geq 4$ and for a skew-symmetric unitary $V\in M_d$ (i.e., $V^{\top}=-V$ and $V^*V=I_d$), let us consider two-parameter linear maps
\begin{equation} \label{eq:SympCov}
    \Le_{p,q}(X):=\frac{1-p-q}{d}\Tr(X)I_d + pX + qVX^{\top} V^*, \quad X\in M_d.
\end{equation}
The family $\{\Le_{p,q}\}_{p,q\in \Real}$ is characterized by the TP maps with symplectic group covariance \cite{Par26}. Furthermore, the family contains the map
    $$\Phi_0=\Le_{\frac{1}{d+2},\frac{1}{d+2}}:X\mapsto \frac{1}{d+2}(\Tr(X)I_d+X+VX^{\top}V^*)$$
which is PPT and shown to be in the class $\SP_{d/2}\setminus \SP_{d/2-1}$ (\cite[Corollary 4.3]{Par26}). We claim that the cone $\DSP_2$ contains all PPT maps in the family $\{\Le_{p,q}\}$, i.e.,  $\Le_{p,q}\in \DSP_2$ whenever $\Le_{p,q}$ is PPT. First, it is shown that the set $\{\Le_{p,q}\}\cap \PPT$ is a convex compact set whose extremal points are precisely $\Phi_0$ and three other maps
    $$\Phi_1= \Le_{\frac{1}{d+1},0}, \quad \Phi_2= \Le_{0,\frac{1}{d+1}}, \quad \Phi_3=\Le_{-\frac{1}{d^2-d-2},-\frac{1}{d^2-d-2}},$$
which are entanglement breaking \cite[Theorem 4.1]{Par26}. Therefore, it suffices to verify $\Phi_0\in \DSP_2$. To this end, let us write $\Phi_0$ into the convex sum $\Phi_0 = \frac{1}{2}(\Le_{p_0,q_0}+\Le_{q_0,p_0})$ where
    $$p_0=\frac{2d-3}{(d+2)(d+1)},\quad q_0 = \frac{5}{(d+2)(d+1)},$$
so that $p_0+q_0=\frac{2}{d+2}$ and $\frac{d-3}{2d-3}p_0+q_0=\frac{1}{d+1}$. Then $(p_0, q_0)$ further satisfies
    $$-\frac{2d-1}{d+1}\leq -(d-1)p_0+q_0\leq \frac{1}{d+1}, \quad p_0-(d-1)q_0\leq 1,$$
and therefore, $\Le_{p_0,q_0}\in \SP_2$ by \cite[Theorem 4.2]{Par26}. On the other hand, since $V$ is a skew-symmetric unitary, one further has
    $$\top \circ \Le_{q_0,p_0} = \Ad_{V^*} \circ \Le_{p_0,q_0} \in \SP_2.$$
Thus we have $\Phi_0\in \DSP_2$.

We remark that the PPT-squared conjecture already holds within the class $\{\Le_{p,q}\}$ \cite{Pru25,Par26}:
    $$\Le_{p,q},\Le_{p',q'}\in \PPT \implies \Le_{p,q}\circ \Le_{p',q'}\in \EB.$$
Theorem \ref{thm:DSP-bridge} newly implies the following strong PPT-cubed assertion
    $$\Le_{p,q},\Psi_1,\Psi_2\in \PPT \implies \Psi_2\circ \Le_{p,q}\circ \Psi_1\in \EB,$$
where $\Psi_1$ and $\Psi_2$ are \emph{arbitrary} PPT maps.
\end{example}

As verified in Examples \ref{ex:DSPex} and \ref{ex:DSP-Symp}, the cone $\DSP_k$ with $k\geq 2$ contains a large family of positive maps and clearly $\PPT\subseteq \DSP_d= \DEC$. We leave open the following question.

\begin{question} \label{ques:DSP_PPT}
What is the smallest integer $k\leq d$ such that
    $$\DSP_k\supseteq \PPT?$$
\end{question}

Theorem \ref{thm:DSP-bridge} implies that the PPT-cubed conjecture would hold on $M_d$ whenever $\DSP_2(d,d)$ contains $\PPT(d,d)$. When $d=3$, it is known that \cite{YLT16} 
    $$\PPT(3,3)\subseteq  \SP_2\subseteq \DSP_2,$$
while it is unknown whether $\PPT(4,4)\subseteq\DSP_2(4,4)$. On the other hand, we can consider the duality between mapping cones: for a cone $\mathcal{C}\subseteq \PO(d_A,d_B)$, we define its \emph{dual cone}
    $$\mathcal{C}^{\circ}:=\{\Psi\in B^h(M_{d_A},M_{d_B}): \Tr(C_{\Psi}C_{\Phi})\geq 0\;\;\forall\,\Phi\in \mathcal{C}\}.$$
If $\mathcal{C}$ is a mapping cone, then so is $\mathcal{C}$. Furthermore, we have $\SP_k^{\circ} = \PO_k$ \cite{TH00,SSZ09}, and therefore,
    $$\DSP_k^{\circ} = \SP_k^{\circ} \cap (\top\circ \SP_k)^{\circ}  = \PO_k\cap (\top \circ \PO_k).$$
Thus, we have the following equivalent formulation of Question \ref{ques:DSP_PPT}.

\begin{proposition} \label{prop:DSP_PPT_equiv}
For all $d_A,d_B\geq 2$ and for each integer $1\leq k\leq \min(d_A,d_B)$, the following are equivalent:
\begin{enumerate}
    \item $\PPT(d_A,d_B)\subseteq \DSP_k(d_A,d_B)$;

    \item $\PO_k(d_A,d_B)\cap (\top\circ \PO_k(d_A,d_B))\subseteq \DEC(d_A,d_B)$. In other words, every linear map $\Phi:M_{d_A}\to M_{d_B}$ which is both $k$-positive and $k$-copositive is decomposable. 
\end{enumerate}
\end{proposition}

\noindent{\bf AI Statement.} The author acknowledges the use of AI tools, including ChatGPT, for language polishing, LaTeX editing, and exploratory mathematical discussions during the development and preparation of this manuscript. Some ideas used in the proof-development process in Section \ref{sec:splitting} arose during interactions with GPT-5.6 Sol and were independently verified, reformulated, and incorporated into the manuscript by the author. The author takes full responsibility for all mathematical content in the final manuscript.

\bigskip

\noindent{\bf Acknowledgments.} The author thanks Aabhas Gulati and Ion Nechita for the helpful discussions and comments. S.-J. Park is supported by the National Natural Science Foundation of China (grant No.~12401163) and the Department of Science and Technology of Hubei Province (Project No.2025EHA041, Project No.2025AFA044). S.-J. Park acknowledges the Hongyi Postdoc Fellowship from Wuhan University.

\bibliography{references}
\bibliographystyle{alpha}
\bigskip
\hrule
\bigskip

\end{document}